\documentclass[conference]{IEEEtran}

\usepackage{cite}
\usepackage{amsmath,amsfonts}
\usepackage{graphicx}
\usepackage{textcomp}
\usepackage{xcolor}
\usepackage{amsmath}
\usepackage{amsthm}
\usepackage{mathtools}
\usepackage{todonotes}
\usepackage{subcaption}
\usepackage{hyperref}
\usepackage{booktabs}
\usepackage{algorithm}
\usepackage[commentColor=blue]{algpseudocodex}
\usepackage{listings}
\usepackage{xcolor}
\usepackage{minted}
\usepackage{hhline}
\usepackage{array}
\usepackage{multirow}
\usepackage{enumitem}

\newtheorem{lemma}{Lemma}

\DeclareMathOperator*{\argmax}{arg\,max}

\def\BibTeX{{\rm B\kern-.05em{\sc i\kern-.025em b}\kern-.08em
    T\kern-.1667em\lower.7ex\hbox{E}\kern-.125emX}}

\begin{document}
\newcommand{\Aydin}[1]{{\color{red}(Aydin) #1}}
\newcommand{\Dmitriy}[1]{{\color{magenta}(Dmitriy) #1}}
\newcommand{\Kathy}[1]{{\color{cyan}(Kathy) #1}}
\title{Distributed-Memory Parallel Algorithms for Fixed-Radius Near Neighbor Graph Construction}

%\author{Anonymous}
\author{
    \IEEEauthorblockN{Gabriel Raulet\IEEEauthorrefmark{1}, Dmitriy Morozov \IEEEauthorrefmark{2}, Ayd\i n Bulu\c{c}\IEEEauthorrefmark{1}\IEEEauthorrefmark{2}, Katherine Yelick \IEEEauthorrefmark{1}}
    \IEEEauthorblockA{\IEEEauthorrefmark{1}\textit{University of California, Berkeley, CA, USA}}
    \IEEEauthorblockA{\IEEEauthorrefmark{2}\textit{Lawrence Berkeley 
    National Laboratory, Berkeley, CA, USA}}
}

%\author{\IEEEauthorblockN{Gabriel Raulet}
%\IEEEauthorblockA{
%University of California, Berkeley\\
%Berkeley, CA, USA \\
%gabe.h.raulet@berkeley.edu}
%\and
%\IEEEauthorblockN{Dmitriy Morozov}
%\IEEEauthorblockA{
%Lawrence Berkeley National Laboratory\\
%Berkeley, CA, USA \\
%dmorozov@lbl.gov}
%\and
%\IEEEauthorblockN{Ayd\i n Bulu\c{c}}
%\IEEEauthorblockA{
%Lawrence Berkeley National Laboratory\\
%Berkeley, CA, USA \\
%abuluc@lbl.gov}
%\and
%\IEEEauthorblockN{Katherine Yelick}
%\IEEEauthorblockA{
%University of California, Berkeley\\
%Berkeley, CA, USA \\
%yelick@berkeley.edu}
%}

\maketitle
\begin{abstract}

Computing fixed-radius near-neighbor graphs is an important first step for many data analysis algorithms. Near-neighbor graphs connect points that are close under some metric, endowing point clouds with a combinatorial structure.
As computing power and data acquisition methods advance, diverse sources of large scientific datasets would greatly benefit from scalable solutions to this common subroutine for downstream analysis.
Prior work on parallel nearest neighbors has made great progress in problems like k-nearest and approximate nearest neighbor search problems, with particular attention on Euclidean spaces. Yet many applications need exact solutions and non-Euclidean metrics.

This paper presents a scalable sparsity-aware distributed memory algorithm using cover trees to compute near-neighbor graphs in general metric spaces. We provide a shared-memory algorithm for cover tree construction and demonstrate its competitiveness with state-of-the-art fixed-radius search data structures. We then introduce two distributed-memory algorithms for the near-neighbor graph problem, a simple point-partitioning strategy and a spatial-partitioning strategy, which leverage the cover tree algorithm on each node. Our algorithms exhibit parallel scaling across a variety of real and synthetic datasets for both traditional and non-traditional metrics. On real world high dimensional datasets with one million points, we achieve speedups up to 678.34x over the state-of-the-art using 1024 cores for graphs with 70 neighbors per vertex (on average), and up to 1590.99x using 4096 cores for graphs with 500 neighbors per vertex (on average).

\end{abstract}
\begin{IEEEkeywords}
Computational geometry, Topological data analysis, Distributed-memory parallelism
\end{IEEEkeywords}

% Some macros for math

% Sections
\section{Introduction}

A wide range of algorithms require near neighbor graphs as a subroutine, including non-linear dimensionality reduction techniques (Isomap, UMap, tSNE, etc.)~\cite{tenenbaum00} \cite{McInnes2018} \cite{vandermaaten08a}, clustering methods (DBSCAN, spectral clustering, HipMCL, FOF) \cite{ester96} \cite{azad18}, density estimation, and topological data analysis (Vietoris-Rips complexes) \cite{zomordian10}.

The basic problem has several formulations. Two of the most common ones are: given a set of $n$ points $P$ in some metric space $(X,d)$, find (1) all pairs of points within some threshold distance $\epsilon$; or (2) the $k$ nearest neighbors of every point in $P$. The choice between the two depends on the downstream algorithm and is always tightly coupled to its analysis. In this paper, we are concerned only with the first setting: a fixed-radius near-neighbor search~\cite{bentley275}.

When $X$ is a Euclidean space, the points are described by their coordinates and graph construction algorithms exploit this additional structure by partitioning the space (for example, k-d trees). But there are many other settings where the metric is not Euclidean---for example, edit distance is used to compare strings (genomic analysis) and Wasserstein distance is used to compare point distributions (statistics and topological data analysis). In this paper, we are concerned with the most general metric setting, where only the triangle inequality can be assumed.

Because the datasets can be large, and distance computations expensive, a robust solution requires us to distribute the computation across many compute nodes. There are different regimes applicable based on the sparsity of the output graph. When the graph is dense and its construction is compute-bound (large $\epsilon$ and/or high intrinsic dimensionality), one can do no better than parallelizing all $n \choose 2$ pairwise distances and pruning those above the threshold. In most interesting real-world applications, however, the output graph is assumed to be sparse. This hypothesis is usually justified by assuming that the input points sample some lower dimensional space, the true ``data manifold'', and it is the intrinsic dimensionality of that space that controls the number of edges, rather than the ambient space.

In this sparse setting, it is possible to build the graph in sub-quadratic time. Many (single-node) data structures have been developed to accelerate this computation (e.g., ball trees, vp-trees, etc.)\cite{omohundro1989balltree} \cite{yianilos93}. For general metric spaces, cover trees are an attractive method both because of their theoretical guarantees and empirical performance. In this paper, we provide an implementation of a shared-memory cover tree optimized for batch construction. We show that our implementation is efficient and competitive with state of the art fixed-radius near neighbor search approaches in a shared-memory environment.

To build the near neighbor graph in distributed-memory, we developed two different algorithms that use cover trees as a data structure to accelerate neighbor queries in a local address space. The first approach is inspired by distributed particle simulators, and uses a simple point-partitioning strategy in conjunction with a ring pipeline to find all pairwise neighbors. The second approach is a spatial partitioning strategy, which constructs a finite Voronoi diagram with a fixed number of centers, and then coalesces the Voronoi cells onto processors. This approach allows more efficient usage of the cover tree and reduces the required number of queries. At higher node counts the communication requirements of the spatial partitioning method (an expensive all-to-all collective) leads to some degradation in scaling. As a result, for some datasets, the point-partitioning strategy performs better. To remedy this, a hybrid approach using the ring pipeline to perform queries of a subset of intra-cell pairs improves the scalability of the spatial-partitioning approach, and in most cases is the overall winner.

\section{Related Work}
\label{sec:related}
Nearest neighbor search is an important and well-studied topic. Most algorithms can be divided into indexing and querying stages. Indexing is done once and is the place where internal data structures, such as trees and graphs, are built. Querying happens for each point using these data structures.

Sparsity in the output of nearest neighbor search is a key property for its efficiency. Sparsity can be enforced by limiting the search to the k nearest neighbors, resulting in the acronym KNN search. 
Xiao and Biros~\cite{xiao2016parallel} give distributed-memory parallel algorithms for KNN search. They present both brute-force algorithms as well as parallelizations of the randomized k-d tree data structure.  

Alternatively, output sparsity can be enforced by limiting the search to a \emph{fixed radius}, which is the problem we solve in this paper. 
The closest work to ours is the SNN search method of Chen and Guttel~\cite{chen2024fast} that solves the exact fixed-radius nearest neighbors problem. SNN algorithm first computes a thin SVD in $O(nd^2)$ time for $n$ points in $d$-dimensional space. The inner product of the first right singular vector, i.e., the principal component, with each point becomes that point's score. Points are then sorted and filtered based on their scores. This completes the indexing phase. Querying uses BLAS3 operations for high performance.  
Gu et al.~\cite{gu2022parallel} present a work-efficient shared-memory parallel algorithms for cover trees, but without an implementation or experimental results. We build on their algorithms in this work.

%Another result we use in our parallel algorithm design is by Jahanseir and Sheehy~\cite{jahanseirconstructing}, who introduce an efficient method to build hierarchical data structures, like cover trees for metric spaces, using \emph{locally greedy permutations}. A greedy permutation is an ordering of points in a metric space where each successive point is chosen to maximize its minimum distance from all previously selected points. This method, also known as a Gonzalez ordering or farthest-point traversal, provides a structured way to sample points while maintaining uniform coverage across the space. A locally greedy permutation is a generalization of greedy permutations, where each point in the permutation is chosen to be $\delta$-close to its nearest predecessor in the existing set. This relaxes the strict greedy requirement while preserving useful geometric properties. 

A growing number of methods solve the ``approximate'' nearest neighbors problem. Fast approximate nearest neighbor (ANN) construction methods have been designed to exploit input characteristics such as ultra high dimensionality~\cite{wang2018randomized}.  A recent promising technique called small-world navigable graphs (NSW)~\cite{malkov2014approximate} is outperformed by tree-based methods for low-to-medium dimensionalities~\cite{naidan2015permutation}. Its state-of-the-art extension, hierarchical NSW~\cite{malkov2018efficient}, while performing better for low-to-medium dimensional datasets, comes with distributed-memory parallelization limitations. This is because the search algorithm of HNSW starts from the top layer, making it hard to distribute due to congestion of the higher layer elements. 

Many other notable methods for ANN search involve randomization, such as those that use locality sensitive hashing (LSH)~\cite{indyk1998approximate}, random projection trees (RPTs)~\cite{dasgupta2008random}, or random rotations~\cite{jones2011randomized}. A recent study on distributed-memory parallel sparse RPTs~\cite{ranawaka2023distributed} showed that there is no clear winner among ANN construction methods, demonstrating the need for more research. 

\section{Definitions}
Our work is concerned with general metric spaces. We formally introduce them here, as well as a few other helpful definitions.
A \textbf{metric space} $(X,d)$ is a set of points $X$ along with a distance function $d : X \times X \to \mathbb{R}$ such that for all $x,y,z \in X,$
\begin{itemize}
    \item[(i)] $d(x,y) \geq 0$ (distances are nonnegative),
    \item[(ii)] $d(x,y) = 0$ if and only if $x = y$ (distance to self is $0$),
    \item[(iii)] $d(x,y) = d(y,x)$ (distances are symmetric), and
    \item[(iv)] $d(x,y) \leq d(x,z) + d(z,y)$ (triangle inequality).
\end{itemize}
Requirement (ii) in the above definition implies that duplicate points do not exist. We cannot usually make this assumption with real world datasets. Duplicate points are automatically identified in the cover tree batch construction algorithm and assigned to a common leaf vertex. This allows us to treat them as single points (associated with multiple graph vertices) without violating the cover tree invariants.
Any finite subset of $P \subset X$ is a \textbf{finite metric space}. Throughout this paper, we assume that our dataset is a finite metric space $P \subset X$ with $n$ points, whose metric is inherited from the larger space $X$. Common distance metrics include the Euclidean $l_2$ norm, cosine distance, hamming distance, and edit distance, for example.
The main goal of this paper is to construct a \textbf{near neighbor graph}. Formally, given a set $P \subset X$ of $n$ points and a threshold parameter $\epsilon > 0$, an $\epsilon$-\textbf{graph} on $P$ is an undirected graph $G = (V,E)$ where $V = P$ and $$E = \{(p,q) : d(p,q) \leq \epsilon\}.$$
For any subset $S \subset P$, the distance between a point $p \in P$ and the set $S$ is the distance between $p$ and its nearest neighbor in $S$. That is, $d(p,S) = \min_{x \in S}d(p,x).$
%The \textbf{diameter} of a set of points $P$ is the maximum pairwise distance between points of $P$. Th \textbf{co-diameter} of $P$ is the minimum pairwise distance between distinct points of $P$.
The \textbf{ball} in $P$ with center $p$ and radius $r \geq 0$ is the set of points $$B(p,r) = \{x \in P : d(p,x) \leq r\}.$$
We denote the set of balls of radius $r$ centered at the points of a set $S \subset P$ by $S_r = \bigcup_{p \in S}B(p,r).$
If $S \subset P$, then the \textbf{coverage radius} of $S$ with respect to $P$ is the minimum value $r \geq 0$ such that $P \subset S_r.$
$S$ is called an \textbf{$r$-net} for $P$ if it has coverage radius at most $r$ and minimum separation at least $r$ (i.e., $d(p,q) \geq r$ for all $p \neq q \in S$).
The \textbf{spread} of a set $P$, denoted by $\Delta(P)$, is the ratio of the largest to the smallest pairwise distances in $P$. The spread is alternatively called the \textbf{aspect ratio} in the literature.

Suppose $C = \{c_1, \ldots, c_m\} \subset P$. A \textbf{Voronoi diagram} is a partitioning of $P$ induced by the points of $C$ such that $$V_i = \{p \in P : d(p,c_i) = d(p,C)\}.$$
Each $c_i$ is called a \textbf{Voronoi site}, or \textbf{center}, and its corresponding set of points $V_i$ is called a \textbf{Voronoi cell}. Note that if two points are the same distance from a given center $c_i$, we only assign one of the points to $V_i$ to avoid double counting.
We define the radius of $V_i$ to be $r_i = \max_{p \in V_i}d(p,c_i).$
Intuitively, a Voronoi cell associated with a center $c_i$ is the set of all points whose closest center is $c_i.$

We say that the finite metric space $P$ has \textbf{expansion constant}\cite{karger02} $c$ if $c \geq 2$ is the smallest value such that for all $p \in P$ and $r > 0$, we have
$$|B(p,2r)| \geq c \cdot |B(p,r)|.$$
The expansion constant formalizes a notion of the intrinsic dimensionality of a general metric space. Metric spaces with high expansion constant are, intuitively, less uniform and more difficult to handle. Lower expansion constants, however, decrease the exponential blowup seen in the aforementioned case.

We now introduce our main data structure. Let $S \subset P$ be a set of $m$ points. A \textbf{cover tree}\cite{beygelzimer06} on $S$ is a tree $T = T(S)$ whose vertices are each associated with a point in $S$. 
The lowest level of the tree contains $m$ leaf vertices, each of which is associated with a unique point in $S$.
Every vertex in the tree is uniquely identifiable by its point in $S$ and its level in $T$. The vertex associated with a point $p$ and level $k$ is denoted by $p^k \in T$. The parent of a vertex $p^k$ is on level $k+1$, and is denoted by $par(p^k)$.
The children of a vertex $p^k$ are on level $k-1$ and are denoted by $ch(p^k).$ 
In addition, a cover tree must fulfill the following invariants:

\begin{itemize}
    \item[(i)] (\textbf{nesting property}) \emph{Every non-leaf vertex must have a nested child associated with the same point as itself};
    \item[(ii)] (\textbf{covering property}) \emph{$ch(p^k) \subset B(p,2^k)$ for all $p^k \in T$}; and
    \item[(iii)] (\textbf{separating property}) \emph{$d(p_1,p_2) \geq 2^k$ for all $p_1 \neq p_2 \in ch(q^{k+1})$}.
\end{itemize}

The stated separating property is a relaxation of the original definition, which required all points on the same level to fulfill the separating condition (as opposed to just siblings). In practice, this has been found to perform comparably to the stricter definition \cite{izbicki15}.
The runtime complexity's of the cover tree are parameterized by the number of points and the expansion constant of the given space. It has been pointed out in the literature that the original paper introducing cover trees gave a time bound of $O(c^{12}\log{n})$, but that the proof had an error\cite{elkin23}. Further work demonstrated a new time bound but required defining a new structure they call a \textbf{compressed cover tree}\cite{elkin22}. We did not implement this structure, and the relaxed separating property removes a condition relied on by these runtime bounds.
That being said, we do make the assumption here that $P$ has low intrinsic dimensionality, for otherwise the cover tree does not help, and we can do no better than brute force. Under this assumption, we simply drop the constant $c$ and obtain worst-case tree construction bounds of $O(n\log{n})$ and single point query bounds of $O(\log{n})$.

%The runtime complexity's of the cover tree are parameterized by the number of points and the expansion constant of the given metric space. It can be shown that the number of children for a given vertex in the tree is bounded by $O(c^4)$.
%It has been proved by ~\cite{beygelzimer06} that cover trees can be constructed in $O(c^6n\log{n})$ time for $n$ points, and point queries can be serviced in $O(c^{12}\log{n})$ time.

\section{Algorithm}

\subsection{Cover Tree Batch Construction}

We provide here a batch construction algorithm that avoids making $n$ consecutive point insertions. The basic idea is that we start with all the points in the dataset and select one to be the root of the tree. We then find the children of the root which form an $r$-net, where $r$ is just greater than half the distance from the tree root to the farthest point in the dataset (to satisfy the separating property). We then partition the points as a Voronoi diagram whose centers are the tree root, and recurse on each center treated as a new subtree root. We treat vertices and their descendants as a \textbf{vertex triple} $(H, p, r)$ with an associated level $l$, meaning that $p^l$ is in the tree, and such that
\begin{itemize}
    \item $H \subset P$ is a subset of points (descendants of $p^l$),
    \item $p \in H$ is the \textbf{root} of $H$, and
    \item $r \equiv \max_{q \in H}d(p,q)$ is the \textbf{radius} associated with $p^l$.
\end{itemize}
Associated with vertex triple is its \textbf{ball} $$B(H) \equiv B(p,r) = \{q \in P : d(p,q) \leq r\}.$$
Note that the ball for $H$ may (and often does) contain points outside of $H$.

Assume we are given a vertex triple $(H,\pi_1,r)$ associated with vertex $\pi_1^l$.
The key primitive we require is a vertex splitting function which selects children $\pi_1, \pi_2, \ldots, \pi_{m} \in H$ that induce child triples $(H_1, \pi_1, r_1), \ldots, (H_m, \pi_m, r_m)$
such that
$$H_i = \{p \in H : d(p,\pi_i) \leq d(p,\pi_j) \text{ for all } i \neq j\},$$
and the following invariants are maintained:
\begin{itemize}
    \item[(1)] \textbf{Covering}. $H \subset \bigcup_{i=1}^{m}B(\pi_i, r/2),$
    \item[(2)] \textbf{Separating}. $d(\pi_i, \pi_j) > r/2$ for all $i \neq j.$
\end{itemize}
We then insert vertices $\pi_1^{l-1}, \pi_2^{l-1}, \ldots, \pi_m^{l-1}$ into $T$ as children of $\pi_1^l$. Note that we always select $\pi_1$ as a child to maintain the nesting invariant.
Our splitting function is outlined in Algorithm~\ref{alg:vertex_split}.
The splitting function is used as a key primitive for building a shared memory cover tree, outlined in Algorithm \ref{alg:tree_construct}, which iterates level by level until there are no more vertex triples left.

\subsection{Cover Tree Batch Queries}
For this paper, we are interested in finding all points withing distance $\epsilon$ of the query point. To do this with the cover tree, we start from the root and iterate down from higher to lower levels. At level $l$, we collect tree vertices $p^l$ within a distance $2^l + \epsilon$ of the query point, discarding the rest. We then take all the children of the vertices we haven't discarded to be the tree vertices we explore at the next level, $l-1$. When we reach a leaf vertex, we simply check if it within a distance $\epsilon$ of the query and, if so, report it as an $\epsilon$-near neighbor. In practice, we use the radius of the vertex triples in Algorithm~\ref{alg:tree_construct} instead of $2^l$, as the vertex triple radius is an upper bound on the distance between a vertex and all its descendant leaves. Single point queries are presented as Algorithm~\ref{alg:query}. In our actual implementation, we amortize the query costs across a batch of queries by collecting tree vertices within the required distance to a batch of points.

\begin{algorithm}
\caption{Vertex splitting algorithm}
\label{alg:vertex_split}
\begin{algorithmic}[1]
\Require
\begin{itemize}[leftmargin=-5ex]
    \item Vertex triple $(H, \pi_1, r)$ for vertex $\pi_1^l$,
    \item Distance array $D$ pre-initialized with $D[p] \gets d(p, \pi_1)$ for all $p \in H$,
    \item $\pi_2 = \argmax_{p \in H}D[p]$ \Comment{Note that $r = D[\pi_2]$}
\end{itemize}
\Ensure
\begin{itemize}[leftmargin=-4ex]
    \item Child vertices $\pi_1^{l-1}, \ldots, \pi_m^{l-1}$,
    \item Vertex triples $(H_1, \pi_1, r_1), \ldots, (H_m, \pi_m, r_m)$,
    \item Distance arrays $D_1, \ldots, D_m$ with $D_i[p] = d(p, \pi_i)$ for all $p \in H_i$
\end{itemize}
\hspace{-8ex}with
\begin{itemize}[leftmargin=-3ex]
    \item[(1)] \textbf{Covering}. $H \subset \bigcup_{i=1}^m B(\pi_i, r/2)$,
    \item[(2)] \textbf{Separating}. $d(\pi_i, \pi_j) > r/2$ for all $i \neq j$.
\end{itemize}
\Function{SplitVertex}{$H$, $D$, $\pi_1$, $\pi_2$, $l$}
    \State $L[p] \gets \pi_1$ for all $p \in H$
    \BeginBox
    \State $m \gets 0$
    \Repeat
        \State $m \gets m + 1$
        \State $r^* \gets 0$
        %\State $D^{*}[p] \gets d(p, \pi_m)$ for all $p \in H$
        \For{$p \in H$}
            \State $d^* \gets d(p, \pi_m)$
            \If {$d^* < D[p]$}
                \State $D[p] \gets d^*$
                \State $L[p] \gets \pi_m$
            \EndIf
            \If {$D[p] > r^*$}
                \State $r^* \gets D[p]$
                \State $\pi_{m+1} \gets p$
            \EndIf
        \EndFor
    \Until{$D[\pi_{m+1}] \leq 2^{l-1}$} \Comment{$\pi_{m+1}$ violates separating condition}
    \EndBox
    \BeginBox
    \For{$i \gets 1..m$}
        \State $H_i \gets \{p \in H : L[p] = \pi_i\}$
        \State $D_i[p] \gets D[p]$ for all $p \in H_i$
        \State $\pi_1(H_i) \gets \argmax_{p \in H_i}D_i[p]$ \Comment{Note that $\pi_1(H_i)$ is not $\pi_1$}
        \State $r_i \gets D_i[\pi_1(H_i)]$
        \State $(H_i, \pi_i, r_i)$ is a vertex triple with distance array $D_i$ for new child vertex $\pi_1(H_i)^{l-1}$
    \EndFor
    \EndBox
\EndFunction
\end{algorithmic}
\end{algorithm}

\begin{algorithm}
\caption{Tree construction algorithm}
\label{alg:tree_construct}
\begin{algorithmic}[1]
\Require
\begin{itemize}[leftmargin=-5ex]
    \item Partially completed cover tree $T$ up to level $l$
    \item $\mathrm{Hubs} = \{H(i), \pi_{1}(i)^l, \pi_{2}(i)^l, r(i)\}_{1 \leq i \leq m}$ where $(H(i), \pi_{1}(i), r(i))$ is the vertex triple for $\pi_{1}(i)^l$
    \item Leaf vertex size $\zeta \geq 1$
\end{itemize}
\Ensure Cover tree $T$ with one more level completed
\Function{BuildLevel}{$T$, $\mathrm{Hubs}$, $\zeta$}
    \For {$i \gets 1..m$}
        \State $\{H_j(i)\}_{1 \leq j \leq m_i}$  \\
        $\gets$ \Call{SplitVertex}{$H(i)$, $D(i), \pi_{1}(i), \pi_{2}(i), l$}
        \State $\mathrm{Hubs} \gets \mathrm{Hubs} \setminus \{H(i)\}$
        \For{$j \gets 1..m_i$}
            \State Insert vertex $\pi_1(H_j(i))^{l-1} = B(\pi_1(H_{j}(i)), r_{j}(i))$ 
            \If{$|H_{j}(i)| > \zeta$}
                \State $\mathrm{Hubs} \gets \mathrm{Hubs} \cup \{H_{j}(i)\}$
            \Else
                \If{$|H_{j}(i)| > 1$}
                    \For {$p \in H_{j}(i)$}
                        \State Insert leaf vertex $B(p,0)$
                    \EndFor
                \EndIf
            \EndIf
        \EndFor
    \EndFor
\EndFunction
\end{algorithmic}
\end{algorithm}

\begin{algorithm}[h]
\caption{Cover tree query on point set $P$ (tree $T$)} 
\label{alg:query}
\begin{algorithmic}[1]
\Require query point $q$ and query radius $\epsilon$
\Ensure set of $\epsilon$-neighbors $B(q, \epsilon) \cap P$
\State $\mathrm{Stack} \gets \{\mathrm{root}\}$ \Comment{$\mathrm{root}$ is the tree root of $T$}
\While {$|\mathrm{Stack}| > 0$}
    \State $u \gets \mathrm{PopBack}(\mathrm{Stack})$
    \If {$u$ is a leaf and $d(q,u) \leq \epsilon$}
        \State $u$ is an $\epsilon$-neighbor of $q$
    \Else
        \For {$v \in \mathrm{children}(u)$}    
            \If {$d(q,v) \leq \mathrm{radius}(v) + \epsilon$}
                \State $\mathrm{PushBack}(\mathrm{Stack}, v)$
            \EndIf
        \EndFor
    \EndIf
    
\EndWhile

\end{algorithmic}
\end{algorithm}

\subsection{Systolic Near Neighbor Graph Construction}
\begin{algorithm}[h]
\caption{Systolic near neighbor graph construction} 
\label{alg:systolic}
\begin{algorithmic}[1]
\Require
\begin{itemize}[leftmargin=-5ex]
    \item $N$ processors
    \item local processor rank is $j$
    \item local point set is $P^{(j)}$
\end{itemize}
\Ensure local $\epsilon$-neighbors $\{E^{(j)} = (p,q) \in P^{(j)} \times P : d(p,q) \leq \epsilon\}$

\State Build local cover tree $T_j = T(P^{(j)})$
\For {$i \gets 0..N/2$}
    \BeginBox
    \State Send $P^{((i+j)\mod N)}$ to processor $(j-1+N) \mod N$
    \State Receive $P^{((i+j+1)\mod N)}$ from processor $(j+1) \mod N$
    \Comment{Overlapped with computation below}
    \EndBox
    \BeginBox
    \For {$p \in P^{((i+j) \mod N)}$}
        \State Query $p$ against $T_j$ to find $B(p,\epsilon)\cap P^{(j)}$
    \EndFor
    \EndBox
\EndFor
\end{algorithmic}
\end{algorithm}

\begin{algorithm}[h]
\caption{Coalesce and query Voronoi cells} 
\label{alg:inside}
\begin{algorithmic}[1]
\Require
\begin{itemize}[leftmargin=-5ex]
    \item $N$ processors
    \item local processor rank is $j$
    \item local point set is $P^{(j)}$
    \item Voronoi centers $C = \{c_1, \ldots, c_m\}$
    \item local Voronoi cell sections $V^{(j)}_1, \ldots, V^{(j)}_m$
\end{itemize}
\Ensure
\begin{itemize}[leftmargin=-4ex]
    \item Cell assignment function $f : C \to \{0,\ldots,N-1\}$
    \item assigned Voronoi centers $C^{(j)} = f^{-1}(j)$
    \item assigned Voronoi cells $\tilde{V}_1, \ldots, \tilde{V}_{m_j}$ where $m_j = |C^{(j)}|$
    \item intra-cell $\epsilon$-neighbors $\{(p,q) \in \tilde{V}_i \times \tilde{V}_i : d(p,q) \leq \epsilon \}$ for $i=1..m_j$
\end{itemize}
\State Compute assignment function $f$ \Comment{Multiway number partitioning}
\BeginBox
\For {$i \gets 1..m$}
    \State Send $V^{(j)}_i$ to processor $f(c_i)$
\EndFor
\For {$k \gets 0..N-1$}
    \For {$c_i \in C^{(j)}$}
        \State Receive $V^{(k)}_i$ from processor $k$
    \EndFor
\EndFor

\Comment{Implemented as a collective MPI Alltoallv}
\EndBox
\For {$c_i \in C^{(j)}$}
    \State Coalesce $\tilde{V}_i \gets \bigcup_{k=0}^{N-1}V^{(k)}_i$
    \State Build $T_i \gets T(\tilde{V_i})$
    \For {$p \in \tilde{V_i}$}
        \State Query $p$ against $T_i$ to find $B(p,\epsilon) \cap \tilde{V}_i$
    \EndFor
\EndFor
\end{algorithmic}
\end{algorithm}
\begin{algorithm}[h]
\caption{Collective ghost queries} 
\label{alg:landmark_coll}
\begin{algorithmic}[1]
\Require
\begin{itemize}[leftmargin=-5ex]
    \item $N$ processors
    \item local processor rank is $j$
    \item local point set is $P^{(j)}$
    \item assigned Voronoi trees $T_1, \ldots, T_{m_j}$
\end{itemize}
\Ensure
\begin{itemize}[leftmargin=-4ex]
    \item assigned ghost points $G^{\epsilon}_1, \ldots, G^{\epsilon}_{m_j}$
    \item intra-cell $\epsilon$-neighbors $\{(p,q) \in \tilde{V}_i \times G^{\epsilon}_i : d(p,q) \leq \epsilon\}$
\end{itemize}
\State Build replication tree $T_{rep} \gets T(C)$
\For{$p \in P^{(j)}$}
    \State Query $p$ against $T_{rep}$ to find $G(p) \gets B(p, d(p,C) + 2\epsilon) \cap C$
    \State $F(p) \gets \{ f(c_k) : c_k \in G(p) \setminus \{c_i\}\}$ \Comment{Assuming $p \in V_i$}
\EndFor
\BeginBox
\For{$p \in P^{(j)}$}
    \For{$k \in F(p)$}
        \State Send $p$ to processor $k$
    \EndFor
\EndFor
\For{$k \gets 0..N-1$}
    \For{$c_i \in C^{(j)}$}
    \State Receive $G^{\epsilon}_i \cap P^{(k)}$ from processor $k$
    \EndFor
\EndFor

\Comment{Implemented as a collective MPI Alltoallv}
\EndBox
\For{$c_i \in C^{(j)}$}
    \State Coalesce $G^\epsilon_i \gets \bigcup_{k=0}^{N-1}(G^\epsilon_i \cap P^{(k)})$
    \For{$p \in G^{\epsilon}_i$}
        \State Query $p$ against $T_i$ to find $B(p,\epsilon) \cap \tilde{V}_i$
    \EndFor
\EndFor

\end{algorithmic}
\end{algorithm}
Our first algorithm for near neighbor graph construction utilizes an approach inspired by the application of systolic arrays to molecular simulations \cite{heller90}.
Let $N$ be the number of processors available to us (in this case, MPI processes). For simplicity, assume that $n$ is a multiple of $N$. At the start of all our algorithms, the points in $P$ are equally partitioned across processors with $P^{(j)}$ denoting the initial subset of $n/N$ points available to processor $j$. 
In the systolic algorithm, each processor $j$ builds a cover tree $T^{(j)} = T(P^{(j)})$ in parallel.

Once completed, the query phase is initiated. Each processor $j$ is responsible for computing its $\epsilon$-neighbors $E^{(j)} = \{(p,q) \in P^{(j)} \times P : d(p,q) \leq \epsilon\}.$ To do so, each processor must query every $q \in P$ against its local cover tree $T^{(j)}.$ Anticipating a problem further down the line, it is clearly not desirable to all-gather the dataset $P$ to every processor. We instead pass chunks $P^{(j)}$ around the processors, and overlap each communication step with a querying step, avoiding the memory blowup.

In more detail, we run $N-1$ steps where on step $i$, processor $j$ sends $P^{((i+j) \mod N)}$ to processor $(j-1+N) \mod N$ and receives $P^{((i+j+1)\mod N)}$ from processor $(j+1)\mod N$. Overlapped with this communication step, processor $j$ queries $P^{((i+j)\mod N)}$ against $T^{(j)}$, thereby finding all $\epsilon$-neighbors of $P^{(j)}$ which intersect $P^{((i+j)\mod N)}$. After $N-1$ rounds, the full $\epsilon$-graph is computed, and we are done. In our case, half of these rounds are redundant since distances are symmetric. We therefore only need $N/2$ rounds to exhaust all query pairs $P^{(i)}$ versus $P^{(j)}$ for $i \leq j.$ The latency of this approach scales linearly with the number of processors. However, in virtually all cases distance computations are the time-limiting primitive, and therefore the parallelism achieved by distributing pairwise distance computations makes up for the latency issues in our scaling experiments. That said, communication times are still higher than desired, especially at lower processor counts. The pseudocode for the systolic algorithm is outlined in Algorithm~\ref{alg:systolic}.

\subsection{Landmarking Near Neighbor Graph Construction}

A major downside of the previous approach is that it ignores the global geometry of the dataset. It is smarter to avoid querying the full dataset against each local tree, and to instead coalesce points that are near each other in metric space, and only query points near enough to each other to be potential $\epsilon$-neighbors. To do this, we use a ``landmarking'' approach. At a high level, our algorithm does the following:

\begin{itemize}
    \item [(1)] \textbf{Build Voronoi diagram}: Select a set of $m$ ``centers'' (or landmarks) $C = \{c_1,\ldots, c_m\} \subset P$ where $m \ll n$ is a parameter which scales with the number of processors. $C$ is shared among all processors via an all-gather, and a distributed Voronoi diagram $V_1 \cup \cdots \cup V_m = P$ is constructed, where $V^{(j)}_i = V_i \cap P^{(j)}$ denotes the set of points in Voronoi cell $i$ on processor $j$.
    \item [(2)] \textbf{Coalesce Voronoi cells}: Compute a statically load-balanced assignment $f : [m] \to [N]$ of cells to processors and coalesce the cells on processors via an all-to-all collective. 
    \item [(3)] \textbf{Query intra-cell points}: Each processor builds a cover tree $T_i = T(V_i)$ on each of its assigned Voronoi cells, and then queries each $V_i$ against $T_i$ to find $\epsilon$-neighbors within each $V_i.$ This and the previous step are outlined in Algorithm~\ref{alg:inside}.
    \item [(4)] \textbf{Find and query ghost points}: For each $p \in P$, find all $V_i$ for which it is possible that $p$ has $\epsilon$-neighbors in $V_i$, but is not itself in $V_i$. We call $p$ an $\epsilon$-\textbf{ghost} for $V_i$ if this is true (see Figure~\ref{fig:ghost_points}). Letting $G^{\epsilon}_i$ denote the set of all $\epsilon$-ghosts of $V_i$, we then query all points in $G^{\epsilon}_i$ against $T_i$, and then we are done. This step is outlined in Algorithm~\ref{alg:landmark_coll}.
\end{itemize}

The first three steps are relatively straightforward, with a few caveats. There are many ways to go about selecting the set of centers $C$. One method is to compute the length-$m$ prefix of a greedy permutation\cite{gonzalez85} on $P$. In this case the center set is an $r$-net for $P$, which may appear desirable. However, in high-dimensional settings, or settings with skewed spatial distributions and/or many duplicate points, the greedy permutation strategy can create a highly imbalanced Voronoi partitioning, in terms of the cell sizes \cite{sheehy20onehop}.
A much more reliable approach, which has outperformed greedy permutations on a vast majority of our experiments, is to choose the centers at random.

To compute the assignment function $f$ for the coalescing phase, there are a few competing factors. Using a simple cyclic distribution works, but is not sufficiently sensitive to the imbalance in Voronoi cell sizes and therefore is undesirable to achieve higher parallelism.
A better method is to solve the \textbf{multiway number partitioning} problem on the cell sizes\cite{graham66}. This optimization algorithm minimizes the difference between the maximum and minimum number of points assigned to a processor. The problem is NP-complete, but a $4/3$-approximation can be computed in $O(m\log{m})$ time. This method improves load balancing and is the one used in this paper.

\subsection{Ghost Points}
To find ghost points, we use the following lemma:
\begin{lemma}
\label{ghostlemma}
If $p \in V_j$ has an $\epsilon$-neighbor in $V_i$, where $i \neq j$, then $d(p,c_i) \leq d(p,c_j) + 2\epsilon.$
\end{lemma}

\begin{proof}
    Let $q \in V_i$ be an $\epsilon$-neighbor of $p \in V_j.$ Then
    \begin{align*}
        d(p,c_i) &\leq d(q,c_i) + d(p,q) \\
                 &\leq d(q,c_j) + d(p,q) \\
                 &\leq d(p,c_j) + 2d(p,q) \\
                 &\leq d(p,c_j) + 2\epsilon,
    \end{align*} 
    by triangle inequality and the fact that $d(q,c_i) \leq d(q,c_j)$ and $d(p,q) \leq \epsilon.$
\end{proof}

Using Lemma~\ref{ghostlemma}, we define the set of $\epsilon$-ghosts for $V_i$ by $$G^{\epsilon}_i = \{p \in P \setminus V_i : d(p,c_i) \leq d(p,C) + 2\epsilon\},$$ since $d(p,C) = d(p,c_j)$ when $p \in V_j.$
There are two ways we handle ghost point queries. The first is to have each processor compute ghost points on its local point set $P^{(j)}$, and communicate them to destinations using an all-to-all collective. A difficulty with this approach is that ghost points can be a large fraction of the total points, leading to an undesirable bottleneck in the scaling performance of the all-to-all.
The second is to make use of the ring pipeline described in the previous section. In this case, we send points around in a ring and determine only which local Voronoi cells each point is a ghost for, and query directly. While less efficient when there are few ghost points, this method avoids the memory and scaling issues of the collective approach.
\subsubsection{Collective Ghost Queries}
In the collective-based regime, we compute the sets $G^{\epsilon}_i$ in parallel, with each processor querying its initial point set $P^{(j)}$ against a ``replication tree'' $T_{rep} = T(C).$ Recall that $C$ is shared by all processors, and that $m \ll n$, therefore the time required to have each processor redundantly build $T(C)$ locally is negligible. For each query of $p \in V^{(j)}$ against $T_{rep}$, we find all $(d(p,C) + 2\epsilon)$-neighbors, where $d(p,C)$ is already known locally from the Voronoi diagram construction phase. For each $c_i \in B(p,d(p,C) + 2\epsilon) \setminus V(p)$, where $V(p)$ denotes the Voronoi cell that $p$ belongs to, $p \in G^{\epsilon}_i$, so we need to send $p$ to the processor responsible for $V_i$ under assignment $f$.

\subsubsection{Ring Ghost Queries}
In the ring-based regime, we pass points around the ring as in the first algorithm we discussed. In this case, on step $i$ processor $j$ queries $P^{((i+j)\mod N)}$ against $T^{(j)} = T(C \cap P^{(j)})$, so that for each $p \in P^{((i+j)\mod N)}$ we find all $c_k \in B(p,d(p,C)+2\epsilon) \cap P^{(j)}.$ We then look up the local tree $T_k$ and query it to find its $\epsilon$-neighbors. 

\begin{figure}
    \centering
    \includegraphics[width=1\linewidth]{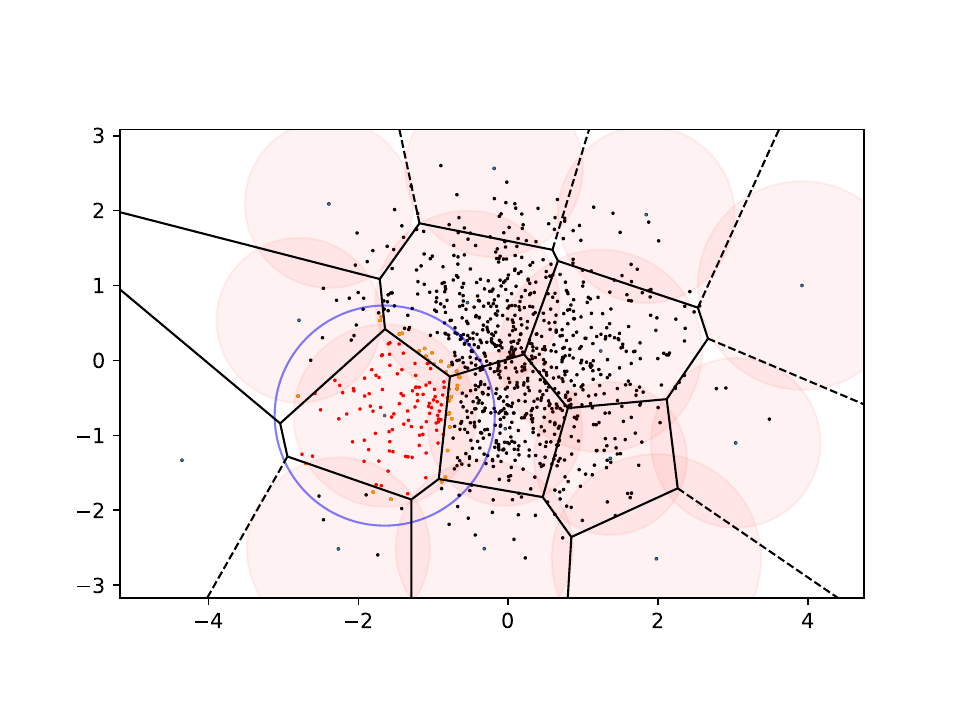}
    \caption{Illustration of $\epsilon$-ghost points in 2D Euclidean space.}
    \label{fig:ghost_points}
    \small
    The convex polygons each correspond to a Voronoi cell. Each red circle represents a ball of radius $r_i$ centered at a Voronoi site $p_i.$ The orange points are $\epsilon$ ghost points for the Voronoi cell whose points are all red, and the blue circle is a ball of radius $r_i + \epsilon$ for the $V_i$ of interest.
\end{figure}

\section{Experiments}

\subsection{Environment}
All experiments were run using NERSC Perlmutter, a Cray HPE EX supercomputer. We used CPU nodes, each of which contains two AMD EPYC 7763 (Milan) processors, with 64 physical cores each. CPU nodes communicate using an HPE Slingshot 11 interconnect. Code was compiled with GCC 13.2.1, using the Cray MPICH implemention of MPI. For all experiments with $128$ MPI ranks or less, we utilized a single compute node. All larger experiments used exactly $128$ MPI ranks per compute node.

\begin{table}[h]
\caption{Datasets}
\centering
\label{tab:datasets}
\resizebox{\columnwidth}{!}{
\begin{tabular}{|cccc|rrr|}

\toprule
Dataset & Metric & Dimension & Points & $\epsilon$-radius & Edges & Avg. neighbors \\
\midrule
\multirow{3}*{\textbf{faces}} & \multirow{3}*{Euclidean} & \multirow{3}*{$20$} & \multirow{3}*{$10304$} & $50$ & $312.7$K & $30.34$ \\
& & & & $100$ & $4.5$M & $436.09$ \\
& & & & $150$ & $17.2$M & $1666.84$ \\
\midrule
\multirow{3}*{\textbf{artificial40}} & \multirow{3}*{Euclidean} & \multirow{3}*{$40$} & \multirow{3}*{$10000$} & $6$ & $112.6$K & $11.26$ \\
& & & & $7$ & $2.5$M & $254.59$ \\
& & & & $8$ & $18.8$M & $1880.145$ \\
\midrule
\multirow{3}*{\textbf{corel}} & \multirow{3}*{Euclidean} & \multirow{3}*{$32$} & \multirow{3}*{$68040$} & $0.1$ & $1.6$M & $24.04$ \\
& & & & $0.125$ & $3.9$M & $57.37$ \\
& & & & $0.15$ & $9.0$M & $132.44$ \\
\midrule
\multirow{3}*{\textbf{deep}} & \multirow{3}*{Euclidean} & \multirow{3}*{$96$} & \multirow{3}*{$10000$} & $0.8$ & $164.1$K & $16.41$ \\
& & & & $1.0$ & $1.4$M & $136.74$ \\
& & & & $1.2$ & $9.6$M & $962.09$ \\
\midrule
\multirow{3}*{\textbf{covtype}} & \multirow{3}*{Euclidean} & \multirow{3}*{$55$} & \multirow{3}*{$581012$} & $150$ & $56.2$M & $96.70$ \\
& & & & $200$ & $157.4$M & $270.85$ \\
& & & & $250$ & $372.9$M & $641.845$ \\
\midrule
\multirow{3}*{\textbf{twitter}} & \multirow{3}*{Euclidean} & \multirow{3}*{$78$} & \multirow{3}*{$583250$} & $2$ & $3.9$M & $6.73$ \\
& & & & $4$ & $34.6$M & $59.29$ \\
& & & & $6$ & $254.3$M & $436.04$ \\
\midrule
\multirow{3}*{\textbf{sift}} & \multirow{3}*{Euclidean} & \multirow{3}*{$128$} & \multirow{3}*{$1000000$} & $125$ & $10.2$M & $10.24$ \\
& & & & $175$ & $71.4$M & $71.41$ \\
& & & & $225$ & $479.9$M & $479.86$ \\
\midrule
\multirow{3}*{\textbf{sift-hamming}} & \multirow{3}*{Hamming} & \multirow{3}*{$256$} & \multirow{3}*{$988258$} & $20$ & $71.1$M & $26.77$ \\
& & & & $30$ & $163.0$M & $164.92$ \\
& & & & $40$ & $648.6$M & $656.29$ \\
\midrule
\multirow{3}*{\textbf{word2bits}} & \multirow{3}*{Hamming} & \multirow{3}*{$800$} & \multirow{3}*{$399000$} & $200$ & $7.7$M & $19.38$ \\
& & & & $250$ & $128$M & $320.68$ \\
& & & & $300$ & $2.1$B & $5186.16$ \\
\bottomrule
\end{tabular}
}
\end{table}
\subsection{Datasets}
All datasets used in our work are listed in Table~\ref{tab:datasets}. This table includes the $\epsilon$ values utilized in our experiments and the corresponding sparsity of the output graph, measured by the average degree of a vertex. Euclidean metric datasets \textbf{faces}, \textbf{artificial40}, \textbf{corel}, \textbf{covtype}, and \textbf{twitter} were sourced from the MLPACK benchmarking suite\cite{curtin13a}. The \textbf{deep} dataset came from NeurIPS '21 challenge on billion-scale approximate nearest neighbor search~\cite{Simhadri2022ResultsOT}. The \textbf{sift} dataset~\cite{jegou11} is our largest Euclidean dataset.
For non-Euclidean results, we used the Hamming metric. The datasets \textbf{sift-hamming} and \textbf{word2bits} were taken from an approximate nearest neighbor benchmarking study~\cite{aumuller20}. As stated in their paper, \textbf{sift-hamming} was generated by embedding the \textbf{sift} dataset into Hamming space using a spatially aware embedding~\cite{lee15}. The \textbf{word2bits} dataset was generated by embedding the top 400,000 words of the 2017 English Wikipedia into an $800$-dimensional Hamming space~\cite{word2bits}.
Three $\epsilon$ values were selected for each dataset. The values chosen were intended to sweep from a super sparse graph with average degree less than $100$ to a denser graph with hundreds or thousands of neighbors on average. The denser graphs (in a relative sense) are noteworthy for they allow more accurate density estimation of geometric spaces, which are difficult to extract from a $k$-nearest neighbor graph.

\subsection{Strong Scaling}

\begin{figure}[hbtp]
\centering
\includegraphics[width=1\linewidth]{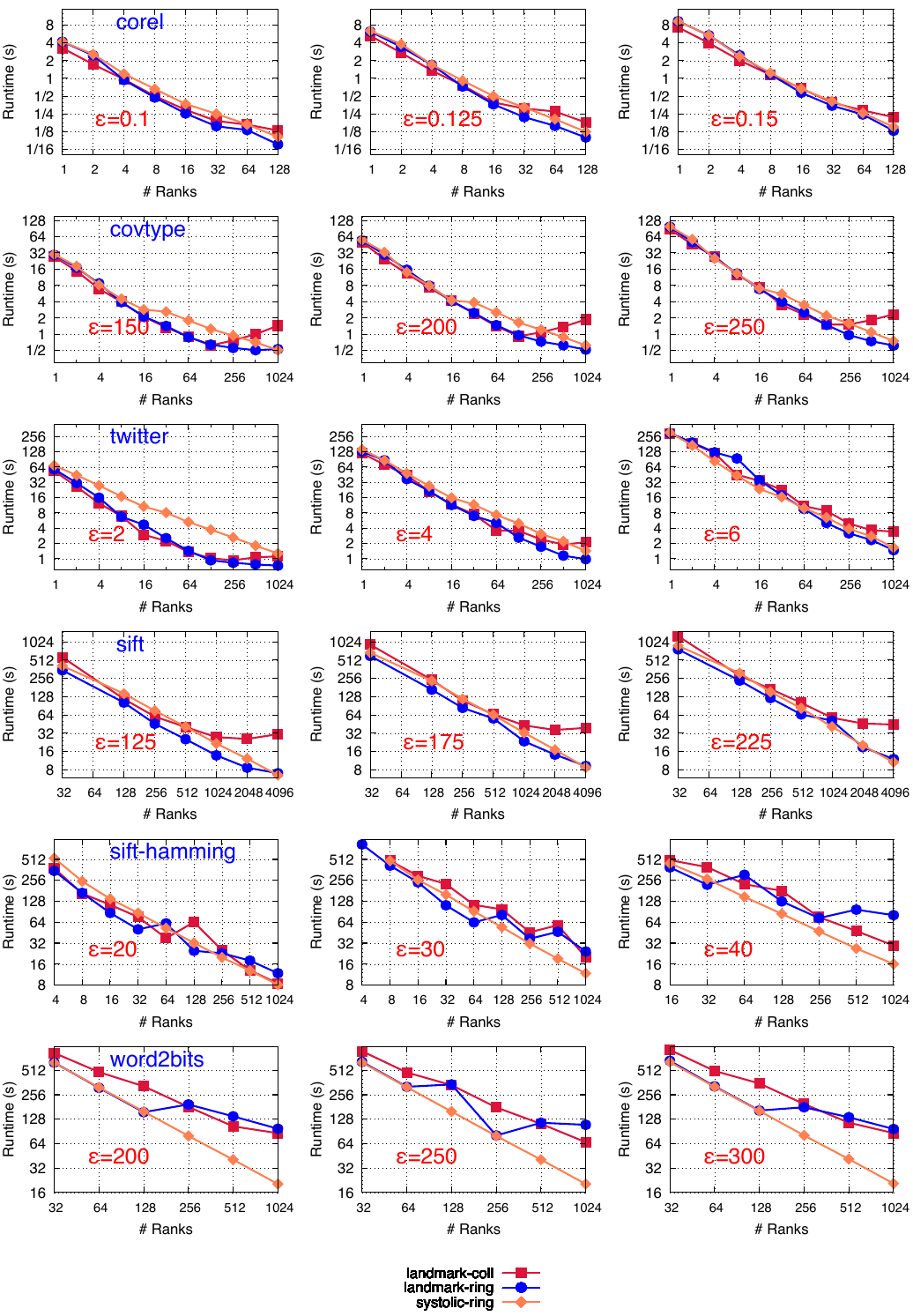}
\caption{Strong Scaling}
\label{scaling}
\end{figure}

For strong scaling experiments, we ran all three algorithms across all datasets and $\epsilon$ parameters. The full results are shown in Figure~\ref{scaling}. The systolic algorithm is labeled \textbf{systolic-ring}, the landmarking algorithm using All-to-all ghost queries is labeled \textbf{landmark-coll}, and the landmarking algorithm using ring ghost queries is labeled \textbf{landmark-ring}. For $\textbf{faces}$, $\textbf{artificial40}$, $\textbf{corel}$, and $\textbf{deep}$ datasets, we used a single compute node, scaling up from $1$ to $128$ processes. Scaling past $1$ node on these datasets is unnecessary for they are all relatively small. For \textbf{covtype} and \textbf{twitter} datasets, we scaled up to $8$ compute nodes (from $1$ to $1024$ processes). For the larger $\textbf{sift}$ dataset ($128$ dimensions and one million points), we ran between $1$ and $32$ compute nodes, with $128$ processes per node. We also ran $32$ processes on a single node. For \textbf{sift-hamming} and \textbf{word2bits}, between $4$-$1024$ and $32-1024$ processors were run, respsectively.

For many of the experiments, \textbf{landmark-coll} starts out superior, but its performance degrades. This can be seen clearly in the \textbf{sift} and \textbf{covtype} experiments, for example. We note that the systolic algorithm tends to eventually converge or beat the best performing landmarking algorithm as we keep increasing the process count. This is expected as all points become part of the ghost regions in high dimensional datasets when process counts increase beyond a certain threshold. However, for many real scenarios, one does not run each dataset in the highest possible concurrency. In those scenarios, landmark bases algorithm provide a good tradeoff between performance and scalability. For example, in the twitter dataset, the systolic algorithm can be an order of magnitute slower than landmark based algorithms in medium process counts. 

\begin{table*}[hbtp]
    \centering
    \caption{
    \textbf{Speedups over SNN}
    This table shows speedups over the sequential SNN at processor counts $N=1,32,1024,4096$. The main takeaway is that the landmarking algorithms are more
efficient at lower processor counts.
    }
    \label{tab:snn_takeover}
    \resizebox{1\textwidth}{!}{
    \begin{tabular}{lcc|ccc|ccc|ccc|ccc}
        \toprule
        \textbf{Dataset} & \textbf{$\epsilon$} & \textbf{SNN} (s) & \multicolumn{3}{c}{\textbf{1-rank speedup}} & \multicolumn{3}{c}{\textbf{32-rank speedup}} & \multicolumn{3}{c}{\textbf{1024-rank speedup}} & \multicolumn{3}{c}{\textbf{4096-rank speedup}} \\
        \midrule
        & & & landmark-coll & landmark-ring & systolic-ring & landmark-coll & landmark-ring & systolic-ring & landmark-coll & landmark-ring & systolic-ring & landmark-coll & landmark-ring & systolic-ring \\
        \midrule
        covtype & 150 & 131.480 & \textbf{4.83} & 4.47 & 4.37 & \textbf{98.34} & 92.43 & 50.36 & 91.16 & 247.84 & \textbf{258.10} & & & \\
        covtype & 200 & 175.037 & \textbf{3.44} & 3.20 & 3.14 & \textbf{72.31} & 70.43 & 45.23 & 93.53 & \textbf{336.93} & 280.94 & & & \\
        covtype & 250 & 217.891 & \textbf{2.50} & 2.24 & 2.20 & \textbf{62.42} & 54.64 & 39.08 & 92.99 & \textbf{352.23} & 289.93 & & & \\
        \midrule
        twitter & 2 &  49.086 & \textbf{0.91} & 0.84 & 0.71 & \textbf{21.89} & 19.29 & 6.06 & 43.97 & \textbf{67.01} & 38.77 & & & \\
        twitter & 4 &  86.739 & \textbf{0.73} & 0.67 & 0.61 & 11.55 & \textbf{12.39} & 7.42 & 40.90 & \textbf{88.46} & 60.99 & & & \\
        twitter & 6 & 125.973 & \textbf{0.43} & 0.42 & 0.41 &  5.50 & 7.05 & \textbf{7.63}  & 36.80 & \textbf{85.38} & 75.86 & & & \\
        \midrule
        sift & 125 & 15374.763 & & & & 26.91 & \textbf{44.21} & 37.37 & 555.33 & \textbf{1106.88} & 708.04 & 496.86 & 2179.99 & \textbf{2339.28} \\
        sift & 175 & 16030.889 & & & & 17.47 & \textbf{26.77} & 23.69 & 372.16 &  \textbf{678.34} & 491.61 & 407.96 & 1727.91 & \textbf{1827.89} \\
        sift & 225 & 17071.052 & & & & 13.48 & \textbf{22.23} & 19.37 & 296.04 &  330.34 & \textbf{415.13} & 383.11 & 1414.85 & \textbf{1590.99} \\
        \bottomrule
    \end{tabular}
    }
\end{table*}

To show the utility of the landmarking algorithms when compute resources are limited, we show in Table~\ref{tab:snn_takeover} the speedups obtained over the SOTA sequential algorithm for our largest Euclidean datasets. This table demonstrates that for $1024$ and fewer processes, the landmarking algorithms are often superior. That said, it is clear from the table that \textbf{systolic-ring} comes out on top for \textbf{sift} at $4096$ processes.

%To test this scenario, we ask the question: ``how many ranks are needed to beat the SOTA sequential algorithm SNN by 5x''. These results are reported in Table\ref{tab:snn_takeover}. 

\begin{figure*}
    \centering
    \includegraphics[width=1\linewidth]{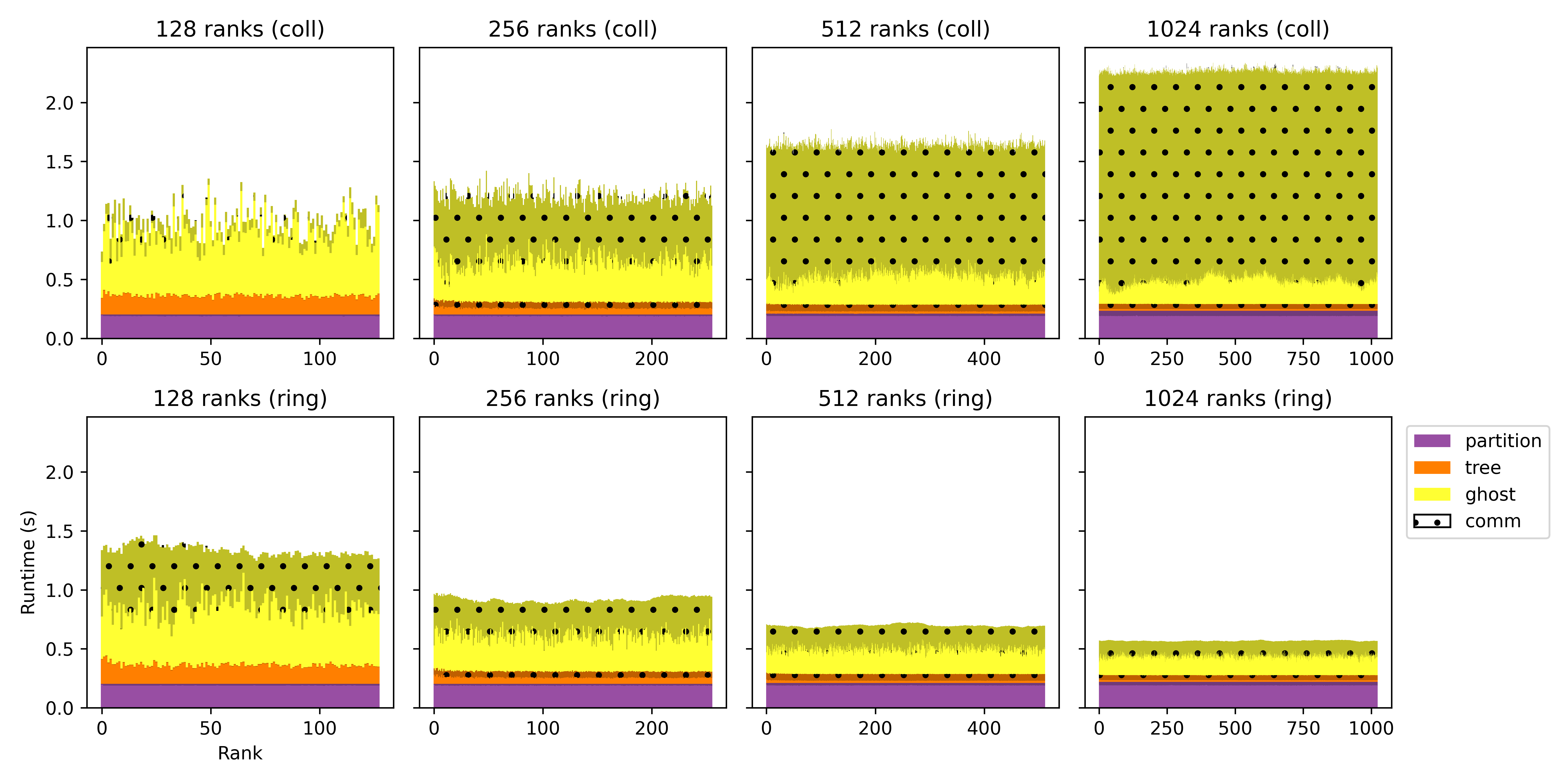}
    \caption{\textbf{covtype} landmark algorithm breakdowns: communication/synchronization times overlayed in darker colors with dotted pattern for each phase, Voronoi partitioning, tree coalescing and querying, and ghost computation and querying. Process counts where communication time begins to dominate the \textbf{landmark-coll} ghost query phase are shown in the top row, and the corresponding results using \textbf{landmark-ring} are shown in second row.}
    \label{fig:covtype_imbalance}
\end{figure*}
\begin{figure*}
    \centering
    \includegraphics[width=1\linewidth]{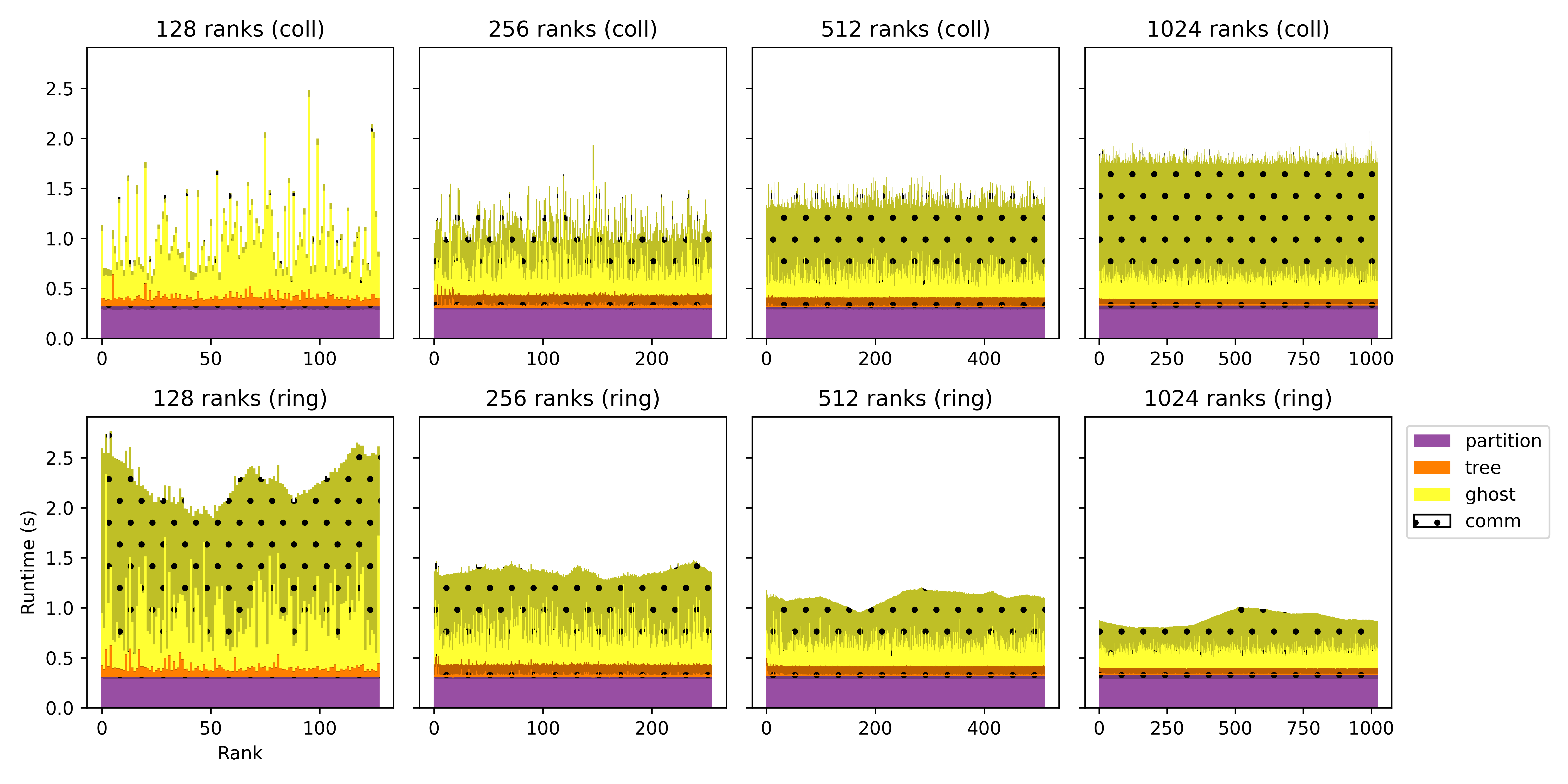}
    \caption{\textbf{twitter} landmark algorithm breakdowns: same presentation as Figure~\ref{fig:covtype_imbalance}.}
    \label{fig:twitter_imbalance}
\end{figure*}
\begin{figure*}
    \centering
    \includegraphics[width=1\linewidth]{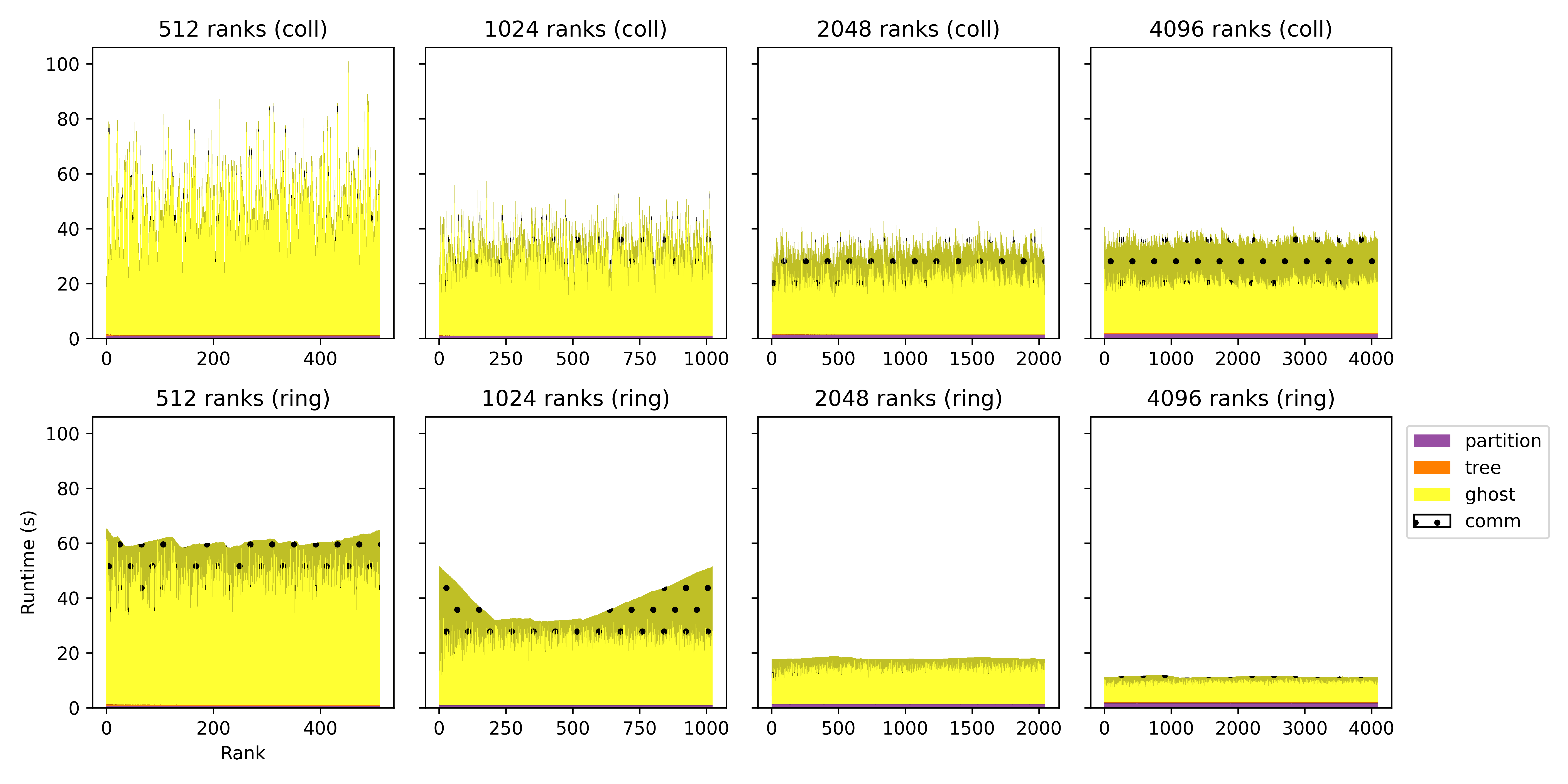}
    \caption{\textbf{sift} landmark algorithm breakdowns: same presentation as Figures~\ref{fig:covtype_imbalance} and~\ref{fig:twitter_imbalance}.}
    \label{fig:sift_imbalance}
\end{figure*}

In Figures \ref{fig:covtype_imbalance}, \ref{fig:twitter_imbalance} and \ref{fig:sift_imbalance}, we present plots that demonstrate the runtimes of the three main phases of our landmarking algorithms: (1) Voronoi partitioning (partition), (2) tree coalescence, construction, and querying (tree), and (3) ghost point determination and querying (ghost). The communication/synchronization time for each phase is highlighted in dark with a dotted background. The runtimes of all ranks are plotted next to each other, making the load imbalance of the algorithm visible. For each figure, the plots on the top use collective ghost queries and the plots on the bottom use ring ghost queries. The main takeaway is that for \textbf{landmark-coll} communication eventually dominates all other phases, motivating the need for a more efficient communication strategy. The \textbf{landmark-ring} algorithm, by contrast, shows that while the communication takes a larger percentage of the runtimes at lower processor counts, they do not blowup as we scale up.

\begin{table*}[hbtp]
    \centering
    \caption{
    \textbf{SNN Direct Comparisons}
    \small
    All entries in this table are runtimes in seconds. SNN was run using its batch query mode. The columns labelled $m=10$ and $m=60$ are runtimes for a single MPI process running the landmark-coll algorithm with $10$ and $60$ Voronoi cells, respectively.}
    \label{tab:direct_comparisons}
    \resizebox{1\textwidth}{!}{
    \begin{tabular}{cccc|cccc|cccc|cccc|cccc|cccc|cccc}
        \toprule
        \multicolumn{4}{c}{\textbf{faces}} & \multicolumn{4}{c}{\textbf{artificial40}} & \multicolumn{4}{c}{\textbf{corel}} & \multicolumn{4}{c}{\textbf{deep}} & \multicolumn{4}{c}{\textbf{covtype}} & \multicolumn{4}{c}{\textbf{twitter}}  & \multicolumn{4}{c}{\textbf{sift}}\\
        \midrule
        $\epsilon$ & SNN & $m=10$ & $m=60$ & $\epsilon$ & SNN & $m=10$ & $m=60$ & $\epsilon$ & SNN  & $m=10$  & $m=60$  & $\epsilon$ & SNN  & $m=10$  & $m=60$  & $\epsilon$ & SNN  & $m=10$  & $m=60$  & $\epsilon$ & SNN  & $m=10$  & $m=60$ & $\epsilon$ & SNN & $m=10$ & $m=60$ \\
        \midrule
         50 &         0.368  & 0.152 & \textbf{0.138} & 6 & \textbf{0.827} & 1.450 & 1.355 & 0.1   &  8.873 & 3.423 & \textbf{2.791} & 0.8 & \textbf{1.294} & 3.879 & 3.562 & 150 & 131.480 & 27.308 & \textbf{25.020} & 2 & 49.086 & 53.887 & \textbf{36.654} & 125 & 15374.763 & \textbf{9679.701} & N/A \\
        100 &         0.510  & 0.439 & \textbf{0.369} & 7 & \textbf{0.933} & 1.521 & 1.487 & 0.125 & 10.488 & 4.971 & \textbf{4.111} & 1.0 & \textbf{1.355} & 4.064 & 3.724 & 200 & 175.037 & 50.339 & \textbf{44.721} & 4 & \textbf{86.739}  & 119.540 & 87.375 & 175 & \textbf{16030.889} & 16932.973 & N/A \\
        150 & \textbf{0.684} & 0.777 & 0.691 & 8 & \textbf{1.123} & 2.173 & 2.089 & 0.15  & 12.005 & 7.086 & \textbf{6.036} & 1.2 & \textbf{1.526} & 4.682 & 4.261 & 250 & 217.891 & 87.617 & \textbf{77.560} & 6 & \textbf{125.973} & 294.492 & 186.772 & 225 & 17071.052 & N/A & N/A \\
        \bottomrule
    \end{tabular}
    }
\end{table*}

%\begin{table*}[hbtp]
%    \centering
%    \caption{
%    \textbf{SNN 5x Speedup}
%    All entries are runtimes in seconds.
%    The SNN/5 column is the SNN runtime divided by $5$ (reference for 5x speedup). The last three columns give the processor count below and above the time at which the given algorithm overtakes the SNN/5 reference.
%    }
%    \label{tab:snn_takeover}
%    \resizebox{0.8\textwidth}{!}{
%    \begin{tabular}{lccccc}
%        \toprule
%        \textbf{Dataset} & \textbf{$\epsilon$} & \textbf{SNN/5} & \textbf{landmark-coll-5x} & \textbf{landmark-ring-5x} & \textbf{systolic-ring-5x} \\
%        \midrule
%        corel & 0.100 & 1.775 & 1--2 & 2--4 & 2--4 \\
%        corel & 0.125 & 2.098 & 2--4 & 2--4 & 2--4 \\
%        corel & 0.150 & 2.401 & 2--4 & 4--8 & 2--4 \\
%        \midrule
%        covtype & 150 & 26.296 & 1--2 & 1--2 & 1--2 \\
%        covtype & 200 & 35.007 & 1--2 & 1--2 & 1--2 \\
%        covtype & 250 & 43.578 & 2--4 & 2--4 & 2--4 \\
%        \midrule
%        twitter & 2 & 9.817 & 4--8 & 4--8 & 16--32 \\
%        twitter & 4 & 17.347 & 8--16 & 8--16 & 16--32 \\
%        twitter & 6 & 25.195 & 16--32 & 16--32 & 8--16 \\
%        \bottomrule
%    \end{tabular}
%    }
%\end{table*}

\subsection{Shared-Memory Cover Tree Performance Comparisons}
One of the contributions of this paper is a software implementation of the cover tree for batch construction and batch querying. Furthermore, our problem of interest is fixed-radius near-neighbor graph construction. Surprisingly, few papers focus on optimizing this specific problem, with most choosing to sacrifice exactness for approximate solutions and/or $k$-nearest neighbor search. The \texttt{radius\_neighbors\_graph} function implemented in scikit-learn\cite{scikit}, uses both kd trees and ball trees. Recent work by \cite{chen2024fast} demonstrated superior performance in both cases in a variety of Euclidean datasets. This work is a useful baseline to compare with our cover tree since it is optimized to solve the same problem we focus on in our paper (fixed-radius search), unlike most other baselines we considered in the nearest neighbor literature.
The direct comparison results with their work (called SNN) are in Table~\ref{tab:direct_comparisons}. We used the landmarking method with collective ghost queries as this paper's representative, run with one rank, as it tends to perform better when there are fewer processes. We stress that SNN requires Euclidean data, as it relies on principal components to sort and filter out potential neighbors provably outside the $\epsilon$-ball of each query. Our work only assumes the metric properties of distance functions.
This enables flexibilty, such as when we want to use very different metrics like hamming distance or edit distance which cannot be optimized via matrix multiplication, but is not the most efficient way to compute Euclidean distances. 
\section{Conclusion}
In this paper we introduced distributed-memory algorithms for construction fixed-radius $\epsilon$-graphs.
We did so using cover trees, for which we implemented efficient batched construction and querying algorithms. We presented a systolic algorithm that shows good scaling, but which is not always the most efficient solution, particularly when (1) there are fewer processes available and (2) the datasets have low intrinsic dimensionality. To handle the latter issues, we implemented a landmarking algorithm using ghost regions, which shows improvements on many datasets at lower processor counts. All of our algorithms are scalable and can be used to build correct near-neighbor graphs with billions of edges, in a short period of time.

\bibliographystyle{IEEEtran}
\bibliography{main.bib}

\end{document}